\documentclass{article}
\usepackage{spconf,amsmath,graphicx}
\usepackage{hyperref}
\def\bA{ {\mathbf{A}} }

\usepackage{graphicx}
\usepackage{amsmath, amsthm, amssymb, amsfonts, cancel}
\usepackage{algorithm, algorithmic, ifsym, subfigure}

\newtheorem{theorem}{Theorem}
\usepackage{epstopdf}
\title{Distributed Probabilistic Bisection Search Using Social Learning}
\twoauthors
 {Athanasios Tsiligkaridis} 
	{Boston University, ECE Dept.\\
	8 St. Mary's St., Boston, MA, USA\\
	atsili@bu.edu}
	{Theodoros Tsiligkaridis} 
	{MIT Lincoln Laboratory\\
	244 Wood St., Lexington, MA, USA\\
	ttsili@ll.mit.edu}
\begin{document}
%\ninept
%
\maketitle
\begin{abstract}
We present a novel distributed probabilistic bisection algorithm using social learning with application to target localization. Each agent in the network first constructs a query about the target based on its local information and obtains a noisy response. Agents then perform a Bayesian update of their beliefs followed by an averaging of the log beliefs over local neighborhoods. This two stage algorithm consisting of repeated querying and averaging runs until convergence. We derive bounds on the rate of convergence of the beliefs at the correct target location. Numerical simulations show that our method outperforms current state of the art methods.%non-Bayesian methods that involves inter-agent collaboration along with the no collaboration case.     
\end{abstract} 
\vspace{-2mm}

\section{Introduction}
\label{sec:intro}

We consider the problem of distributed bisection search using a network of agents. This problem has applications to stochastic root finding \cite{Waeber:2011}, distributed group testing \cite{Han:2016}, object tracking using cameras \cite{Sznitman:2010}, and human-aided localization \cite{Tsiligkaridis:2014}. The agents are connected through a topology and sequentially search for the unknown target $X^*$. Each agent forms a query $Z_{i,t} = I(X^* \in S_{i,t})$ for some subset $S_{i,t} \subset \mathcal{X}$ based on its local information about the target location and obtains a noisy response, $Y_{i,t+1}$, which it uses to update its belief. After this stage, the agents average their log-beliefs with their neighbors. As this process is repeated after several iterations, the agents converge to the correct consensus. %can communicate to the point where all agents can converge to the belief that the target is at a specific location.   

% prior work
Prior work on distributed signal processing includes the consensus literature \cite{Boyd:2006, Dimakis:2010}. Extensions to consensus-plus-innovations algorithms with applications to detection and estimation problems include nodes making noisy observations and implementing a consensus method to spread information around the network. Several such works on distributed estimation include that of \cite{Kar:2011, Kar:2012} which show that distributed estimators are consistent and asymptotically normal. Additionally, \cite{Lalitha:2014, Nedic:2015} studied the problem of distributed detection and proved convergence rates on the rate of learning the correct hypothesis.

The adaptive querying and Bayesian updating are known in the literature as the Probabilistic Bisection Algorithm (PBA). The PBA was first introduced by Horstein \cite{Horstein:1963}. The PBA was shown to be optimal in an information-theoretic sense in the work by Jedynak, et al. \cite{Jedynak:2012}, and the convergence rate of the single-agent PBA was shown to be exponential in Waeber, et al, \cite{Waeber:2013}. The PBA was generalized in \cite{Tsiligkaridis:2014} to multiple agents using a centralized controller strategy, and in \cite{Tsiligkaridis:2015,Tsiligkaridis:2016} using a decentralized algorithm based on belief consensus. In \cite{Tsiligkaridis:2014,Tsiligkaridis:2015,Tsiligkaridis:2016}, the convergence analysis showed that all agents reach consensus to the true target.

% Contributions
The contribution of this paper is to propose and analyze a distributed bisection algorithm using a social learning approach. First, we first derive a social learning algorithm (inspired by the work of Lalitha, et al, \cite{Lalitha:2014}). Second, we derive an asymptotic performance bound that characterizes the rate of learning in terms of the social influence of nodes on the network and the error probability of each agent. Finally, we show using simulations, that our proposed algorithm outperforms the case of no collaboration and improves the distributed bisection search algorithm presented in Tsiligkaridis, et al, \cite{Tsiligkaridis:2015}.

\section{Problem Formulation} \label{sec:BA}

For concreteness, we focus on the one-dimensional case, i.e, $\mathcal{X}=[0,1]$. In this case, the query sets $S_{i,t}$ take the form of an interval $[0,\hat{X}_{i,t}]$, where $\hat{X}_{i,t}$ are the query points. The response to the query $Z_{i,t}=I(X^* \in S_{i,t})$ is modeled as a binary symmetric channel \cite{Jedynak:2012, Tsiligkaridis:2014, Tsiligkaridis:2015} and is given by:
\begin{equation*}
Y_{i,t+1} =  \begin{cases} 
      Z_{i,t} & \text{w.p. } 1 - \epsilon_i \\    
      1 - Z_{i,t} & \text{w.p. } \epsilon_i 
   \end{cases}
\end{equation*}
In this paper, the error probabilities are constant.  In \cite{Tsiligkaridis:2014}, a more general model was considered where the error probability increased as the target localization error decreased. 

We define $p_{i,t}(x)$ as the posterior distribution on the target space $\mathcal{X} = [0,1]$ for agent $i$ at time $t$.  We also denote the corresponding CDF as $F_{i,t}(x) = \int_0^x p_{i,t}(u) du$. The posterior distributions at time $t$ are measurable with respect to the noisy responses $\{Y_{i,\tau+1}\}_{\tau\leq t}$ up to time $t$. The proposed social learning algorithm consists of two stages (see Algorithm $1$). In the first stage, agents perform a probabilistic bisection and update their beliefs using the noisy response. In the second stage, the log-beliefs are averaged among local neighbors to spread information around the network.

The observation probability mass function used in the first stage is given by:
\begin{equation} \label{eq:li}
	l_i(x,y_{i,t+1}) = f_1^{(i)} (y) I (x \leq \hat{X}_{i,t}) + f_0^{(i)} (y) I (x > \hat{X}_{i,t}) 
\end{equation}
where $f_1^{(i)}(y)=(1-\epsilon_i)^{I(y=1)} \epsilon_i^{I(y=0)}, f_0^{(i)}(y)=1-f_1^{(i)}(y)$. \footnote{We remark that the normalizing factor here $\int_0^1 p_{i,t}(x) l_i(x,y_{i,t+1}) dx$ is equal to $1/2$ \cite{Tsiligkaridis:2015}.}

In the second stage, we assume that the social interaction matrix $\textbf{A}$ is a stochastic matrix corresponding to a strongly connected, aperiodic graph $G=(\mathcal{N},E)$.  Here, $\mathcal{N}$ represents the nodes of the network and $E$ is the set of edges.  An edge $(i,j) \in E \text{ exists iff } A_{i,j} > 0$.

\begin{algorithm}
\caption{ Distributed Probabilistic Bisection Algorithm } \label{alg:alg1}
\begin{algorithmic}[1]
\STATE \textbf{Input:}  {$G=(\mathcal{N},E), \bA=\{A_{i,j}\}, \{\epsilon_i\}$}
\STATE \textbf{Output:} {$\{\hat{X}_{i,t}:i\in \mathcal{N}\}$}

    \STATE Initialize $p_{i,0}(\cdot)$ to be positive everywhere.
   
    \REPEAT
        \STATE \underline{Stage 1}: For each agent $i\in \mathcal{N}$: \\
				\STATE  \quad Bisect posterior density to obtain query point $\hat{X}_{i,t} = F_{i,t}^{-1}(\frac{1}{2})$, and obtain noisy response $y_{i,t+1} \in \{0,1\}$. \\
				\STATE  \quad Belief update: \\
        \begin{equation} \label{eq:bayes_update}
            \tilde{p}_{i,t}(x) = p_{i,t}(x) \cdot 2l_{i}(x,y_{i,t+1}) 
        \end{equation}
        \STATE \underline{Stage 2}: For each agent $i\in \mathcal{N}$: \\
        \STATE  \quad Average log-beliefs: \\
        \begin{equation} \label{eq:avg_log_beliefs}
            %p_{i,t+1}(x) = \frac{2^{\left( \sum_{j=1}^N A_{i,j} \log_2 \tilde{p}_{j,t}(x) \right)}}{\int_0^1 2^{\left( \sum_{j=1}^N A_{i,j} \log_2 \tilde{p}_{j,t}(x) \right)} dx}
						p_{i,t+1}(x) = \frac{\prod_{j=1}^{N} \tilde{p}_{j,t}(x)^{A_{i,j}}}{\int_0^1 \prod_{j=1}^{N} \tilde{p}_{j,t}(x)^{A_{i,j}} dx}
        \end{equation}			
    \UNTIL {convergence}
\end{algorithmic}
\end{algorithm}

\vspace{-3mm}

\section{Performance Analysis}
\label{sec:NA}
In this section, we derive a lower bound on the posterior distribution $p_{i,t}(x)$ at the target $X^*$.  Our result shows that for large $t$, we have $p_{i,t}(X^*) \geq 2^{tK}$ for all agents $i \in \mathcal{N}$.
\begin{theorem}
Assume that $\bA$ corresponds to an irreducible, aperiodic Markov chain, and $p_{i,0}(X^*)>0$ for all $i$. Then, the following asymptotic result holds:
\begin{equation}
	\liminf_{t \to \infty} \frac{\log_2 p_{i,t}(X^*)}{t} \geq \sum_{i=1}^N v_i C(\epsilon_i) = K
\end{equation}
where $\mathbf{v}$ is the normalized left eigenvector of $\bA$ which corresponds to the unit eigenvalue \footnote{This is also known as the stationary distribution of the Markov chain.} and $C(\epsilon_i)$ is the capacity of the binary symmetric channel \cite{Cover:2006}.
\end{theorem}

\begin{proof}
Consider the case for a fixed node $i$. From Algorithm \ref{alg:alg1}, we may rewrite (\ref{eq:avg_log_beliefs}) as:
\begin{align*}
\log_{2} p_{i,t+1}(x) &= \sum_j A_{i,j} \log_{2} \tilde{p}_{j,t}(x)  \\
      & - \underbrace{\log_{2} \left( \int_{0}^1 2^{\sum_{j=1}^N A_{i,j} \log_{2} \tilde{p}_{j,t}(x) } dx \right)}_{D_t} 
\end{align*}

Using Jensen's inequality, we have:
\begin{align*}
D_t &\leq \log_2 \left( \int_0^1 \sum_{j=1}^N A_{i,j} \tilde{p}_{j,t}(x) dx \right) = \log_2 \sum_{j=1}^N A_{i,j} = 0
\end{align*}
%\begin{equation*}
	%\log_2 \left( \int_0^1 \sum_j A_{i,j} \tilde{p}_{j,t}(x) dx \right) \leq 0
%\end{equation*}

Since $D_t \leq 0$, it follows that:
\begin{align}
\log_{2} p_{i,t+1}(x) \geq \sum_{j=1}^N A_{i,j} \log_{2} \tilde{p}_{j,t}(x) 
     % &= \sum_{j=1}^N A_{i,j} \log_{2} p_{j,t} (x) + \sum_{j=1}^N A_{i,j} \log_{2} (2l_j(x,y_{j,t+1})) \nonumber
\end{align}

Evaluating the above inequality at $x = X^*$ and using Equation (\ref{eq:bayes_update}) along with $l_j(X^*,y_{j,t+1}) = P(y_{j,t+1}|z_{j,t})$:
\begin{align} \label{eq:EQ6} \hspace{-4mm}
\log_{2} p_{i,t+1}(X^*) &\geq \sum_{j=1}^N A_{i,j} \log_{2} p_{j,t}(X^*) + \sum_{j=1}^N A_{i,j} \log_{2} (2P(y_{j,t+1}|z_{j,t})) %\label{eq:EQ6}
\end{align}

By using new variables:  $q_{i,t} \stackrel{\text{def}}{=} \log_{2} (p_{i,t} (X^*))$ and $w_{j,t} \stackrel{\text{def}}{=} \log_{2} (2P(y_{i,t+1}|z_{i,t}))$, Equation (\ref{eq:EQ6}) becomes:
\begin{equation}
q_{i,t+1} \geq \sum_{j=1}^N A_{i,j} q_{j,t} + \sum_{j=1}^N A_{i,j} w_{j,t}
\end{equation}

Using induction, it follows that:
\begin{equation}
q_{i,t} \geq \sum_{\tau = 1}^t \sum_{j=1}^N A_{i,j}^{\tau} w_{j,t-\tau} + \sum_{j=1}^N A_{i,j}^t q_{j,0}
\end{equation}

Dividing both sides by $t$ and taking the limit, it follows that:

\begin{align} \label{eq:eqNew}
	\liminf_{t \to \infty} \frac{q_{i,t}}{t} &\geq \lim_{t \to \infty} \frac{1}{t} \sum_{\tau = 1}^t \sum_{j=1}^N A_{i,j}^{\tau} w_{j,t-\tau} %\nonumber \\
	  + \lim_{t \to \infty} \frac{1}{t} \sum_{j=1}^N A_{i,j}^t q_{j,0} 
\end{align}

Using $0 < p_{j,0}(X^*) < \infty$ and $A_{i,j}^{t} \leq 1$, the second term in Equation (\ref{eq:eqNew}) vanishes and we have: 

\begin{equation} \label{eq:Ki}
	\liminf_{t \to \infty} \frac{q_{i,t}}{t} \geq \lim_{t \to \infty} \frac{1}{t} \sum_{\tau = 1}^t \sum_{j = 1}^N A_{i,j}^{\tau} w_{j,t - \tau} =: K_i
\end{equation}

Next, we evaluate the rate exponent $K_i$ using the convergence to a stationary distribution $\mathbf{v}$ for aperiodic Markov chains. Specifically, using Theorem 7 in Sec. 2.7 from \cite{Hoel:1972}:
\begin{equation*}
	\lim_{t\to\infty} A_{i,j}^t = v_j
\end{equation*}
for any two nodes $(i,j)$. Thus, for any arbitrarily small $\delta>0$, there exists $T=T(\delta)$ such that $|A_{i,j}^t - v_j|\leq \delta$ for all $t\geq T$.

Splitting the sum in (\ref{eq:Ki}):
\begin{align} \label{eq:SplitEq}
 K_i  &= \lim_{t \to \infty} \Bigg[ \frac{1}{t} \sum_{j=1}^N \sum_{\tau = 1}^{T-1} A_{i,j}^{\tau} w_{j,t - \tau} + \frac{1}{t} \sum_{j=1}^N \sum_{\tau = T}^{t} A_{i,j}^{\tau} w_{j,t - \tau} \Bigg] 
\end{align}
The first term in Equation (\ref{eq:SplitEq}) is negligible since it can be upper bounded by $\lim_{t \to \infty} \frac{1}{t} \sum_{j = 1}^N \sum_{\tau = 1}^{T-1} |w_{j,t - \tau}|$.  Since $|w_{j,t-\tau}| \leq \text{max} \{|\log_2(2(1-\epsilon_j))|,|\log_2(2\epsilon_j)|\} = B_j$, $\lim_{t \to \infty} \frac{1}{t} \sum_{j = 1}^N \sum_{\tau = 1}^{T-1} |w_{j,t - \tau}| \leq \lim_{t \to \infty} \frac{1}{t} (T-1) B = 0$, where $\sum_{j = 1}^N B_j = B$. The second term in Equation (\ref{eq:SplitEq}) dominates and is given by:

\begin{align} \label{eq:eqfin} \hspace{-6mm}
K_i = \lim_{t\to\infty} \frac{1}{t} \sum_{j=1}^{N} \sum_{\tau=T}^{t} [A_{i,j}^{\tau} - v_j] w_{j,t-\tau} + \lim_{t\to\infty} \frac{1}{t} \sum_{j=1}^{N} \sum_{\tau=T}^{t} v_j w_{j,t-\tau} 
\end{align}
The first term in Equation (\ref{eq:eqfin}) is negligible since:
\begin{align}
&\lim_{t\to\infty} \frac{1}{t} \sum_{j=1}^{N} \sum_{\tau=T}^{t} |A_{i,j}^{\tau} - v_j| |w_{j,t-\tau}| \nonumber \\
&\leq \lim_{t\to\infty} \frac{\delta}{t} \sum_{j=1}^{N} \sum_{\tau=T}^{t} |w_{j,t-\tau}| \leq \lim_{t\to\infty} \frac{\delta}{t} \sum_{\tau=T}^{t} \sum_{j=1}^{N} B_j  = \delta B
\end{align}
The second term in Equation (\ref{eq:eqfin}) dominates and using the LLN:
\begin{align} \label{eq:lbound_collab}
	K_i &= \sum_{j=1}^{N} v_j \Bigg[ \lim_{t\to\infty} \frac{1}{t} \sum_{\tau=T}^{t} w_{j,t-\tau} \Bigg] \nonumber \\
		  &= \sum_{j=1}^{N} v_j \mathbb{E}[w_{j,t-\tau}] = \sum_{j=1}^{N} v_j C(\epsilon_j) =: K
\end{align}

The proof is complete.

\end{proof}

In Theorem $1$, we derived a lower bound on the rate of learning of each agent for the distributed bisection algorithm (Algorithm $1$).  Since the rate exponent, $K_i$, is independent of the agent index, $i$, the lower bound is the same for all agents and depends on the channel capacities and the eigenvector centrality, $\textbf{v}$ \cite{Lalitha:2014}.  The higher the value of $v_i$, the larger the contribution that node $i$ has on the network learning rate, $K$.  

Using an analogous method, we can analyze Equation (\ref{eq:bayes_update}) for the case of no inter-agent collaboration. In this case, $\bA = \textbf{I}_N$, and using Equations (\ref{eq:bayes_update}) and (\ref{eq:avg_log_beliefs}), $p_{i,t+1}(X^*) = \tilde{p}_{i,t}(X^*) = p_{i,t}(X^*) + 2P(y_{i,t+1}|z_{i,t})$.  Unrolling the preceding equation over $t$ steps, we see that $p_{i,t}(x) = \prod_{j=0}^{t-1} 2P(y_{i,j+1}|z_{i,j})$.  By taking the logarithm of both sides, dividing by $t$, and applying the LLN, the convergence rate of the agent beliefs becomes: $\lim_{t \to \infty} \frac{\log_{2} p_{i,t}(x)}{t} = C(\epsilon_i)$.  For the case of a homogeneous network, i.e. $\epsilon_i = \epsilon, \forall i$, our distributed algorithm asymptotically learns faster than when agents do not collaborate.  In general, the rate exponent of Theorem $1$ is a linear combination of agents' capacities weighed by the eigenvector centrality. This bound is expected to be maximized by placing nodes with low error probability in central locations in networks where they can have a strong social influence on other nodes.

\section{Simulations}
\label{sec:sim}

To validate and strengthen the preceding performance analysis, simulations are performed that show the proposed method achieves better performance than the belief consensus approach of \cite{Tsiligkaridis:2015} and the case of no inter-agent collaboration. We note that the computational complexity of our proposed algorithm is on the same order as that of \cite{Tsiligkaridis:2015} as both algorithms consist of a Bayesian bisection update and a local averaging operation.  

     We randomly generate an irreducible $N \times N$ adjacency matrix, modeling a random geometric graph \cite{Gupta:2000} (see Figure \ref{fig:network}). In this setup, $N = 20$ and $18$ agents have high error probabilities $(\epsilon_{i} = 0.40)$ while the remaining $2$ have low error probabilities $(\epsilon_{i} = 0.05)$.  The low error agents are the ones with the highest two $v_i$'s, so they can positively affect the high error nodes around them.  Figure \ref{fig:network} displays the 2 low error nodes and their connections in blue along with the other nodes and their connections in black. % \\

     Figure \ref{fig:mse} shows the average network Mean Squared Error (MSE), $MSE_{avg}(t) = \frac{1}{N} \sum_{i=1}^N ( \hat{X}_{i,t} - X^*)^2$, and the worst-case network MSE, $MSE_{max}(t) = \max_i ( \hat{X}_{i,t} - X^*)^2$, averaged over $150$ Monte Carlo trials.  Our proposed method outperforms the belief consensus approach of \cite{Tsiligkaridis:2015} and the case of no collaboration. Regarding the average MSE, all three methods converge to the consensus but our proposed method does so much quicker than the other methods.  In terms of worst case MSE, our proposed method remains robust while that of \cite{Tsiligkaridis:2015} converges slower and the method with no collaboration has not converged after 75 time steps. From this, we see that the proposed method performs very well even in the worst case and that inter-agent collaboration is advantageous in our target localization task. %\\ 

     In addition to the MSE, we also plot $\log_2 p_{4,t}(X^*)$ for both the proposed method and that of no collaboration. The $\log_2 p_{4,t}(X^*)$ time series represent a base-2 logarithm of agent $4$'s belief, an agent with a high error probability. We analytically showed that for the case of no collaboration, a graph of the $\log_2 p_{i,t}(x^*) \text{ vs. } t$ would be a line with slope $K = C(\epsilon_i)$ that would serve as an exact bound; with our proposed method, the same graph should be a line with the slope at least as large as the lower bound $K = \sum_{j=1}^N v_j C(\epsilon_{i})$ (\ref{eq:lbound_collab}).  The analysis results are affirmed by the simulation results in Figure \ref{fig:log}.  Here, the dotted lines represent the two preceding bounds and the solid lines represent the calculations of $\log_2 p_{4,t}(X^*)$ for each iteration.  The solid red line represents the no collaboration method and it lines up with the exact bound displayed in the dotted red line.  The solid blue line represents our proposed method and its slope exceeds that displayed by the lower bound with the dotted blue line.  The simulations verify our analysis and show that the beliefs generated by our proposed method concentrate fast on the true target location. 

		With our distributed algorithm, it is important to note that we improve the network-wide performance (average and worst case) with a small penalty in performance in the case of no collaboration for the lowest error agent. 

%\begin{figure}
    %\centering
    %\includegraphics[width=0.5\textwidth, height=0.4\textwidth]{icassp_plot_1.eps} %width=90mm,height=65mm
		%\vspace{-13mm}
    %\caption{Geometric random graph. The blue nodes are the low error agents and those in red are the high error agents.}
		%\label{fig:network}
%\end{figure} 
%
%\begin{figure}
    %\centering
    %\includegraphics[width=0.5\textwidth, height=0.4\textwidth]{icassp_plot_2.eps}
		%\vspace{-10mm}
    %\caption{Average (top) and worst-case (bottom) MSE performance for our proposed social learning algorithm, the belief consensus approach of \cite{Tsiligkaridis:2015}, and the case of no collaboration.}
		%\label{fig:mse}
%\end{figure} 
%
%\begin{figure}
    %\centering
    %\includegraphics[width=0.5\textwidth, height=0.4\textwidth]{icassp_plot_3.eps}
		%\vspace{-9mm}
    %\caption{Concentration of posterior distributions as a function of iteration.}
		%\label{fig:log}
%\end{figure} 

%\pagebreak
%\vspace{-15mm}
\section{Conclusion}
%\label{sec:refs}
%\vspace{-10mm}
\quad In this paper, we proposed a new distributed probabilistic bisection algorithm for target localization and derived a lower bound on the rate of convergence.  Through analysis and simulation, we show that our proposed method attains superior performance to other state of the art methods in terms of rate of convergence and MSE.  For future work, we can pursue an analysis of the MSE convergence rate.

%\vspace{-10mm}
\begin{figure}[h]
    \centering
    \includegraphics[width=0.54\textwidth]{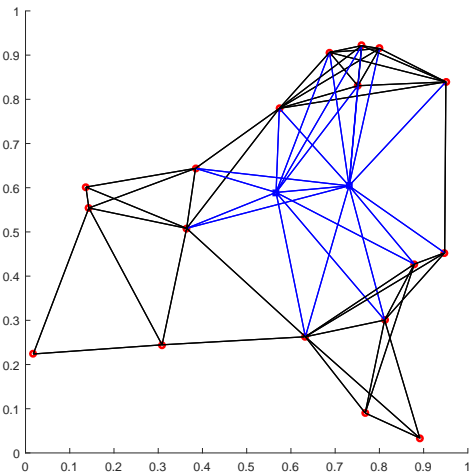} %width=90mm,height=65mm
		\vspace{-6mm}
    \caption{\footnotesize{Geometric random graph. The blue nodes are the low error agents and those in red are the high error agents.}}
		\label{fig:network}
\end{figure} 

\begin{figure}[h]
    \centering
		%\vspace{35mm}
    \includegraphics[width=0.5\textwidth]{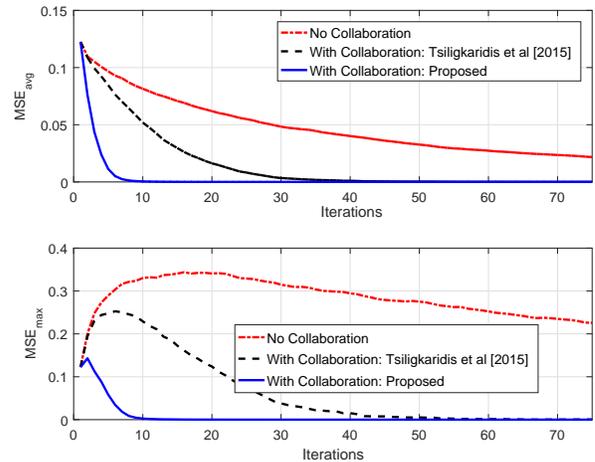}
		\vspace{-5mm}
    \caption{\footnotesize{Average (top) and worst-case (bottom) MSE performance for our proposed social learning algorithm, the belief consensus approach of \cite{Tsiligkaridis:2015}, and the case of no collaboration.}}
		\label{fig:mse}
\end{figure} 

\begin{figure}[h]
    \centering
		%\vspace{70mm}
    \includegraphics[width=0.5\textwidth]{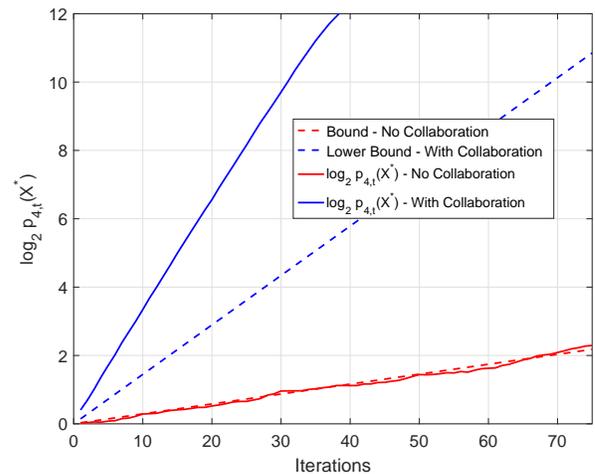}
		\vspace{-3mm}
    \caption{\footnotesize{Concentration of beliefs as a function of iteration. The belief of the high error node, agent 4, of our proposed method concentrates to the target location with a slope of $0.32$ as opposed to the no collaboration case with the slope $0.03$ for the same node. Our method provides $10\times$ improvement over the case of no collaboration regarding the worst case scenario.}}
		\label{fig:log}
\end{figure}

\clearpage
% -------------------------------------------------------------------------
\bibliographystyle{IEEEbib}
\bibliography{refs}

\begin{thebibliography}{10}

\bibitem{Waeber:2011}
R.~Waeber, P.~I. Frazier, and S.~G. Henderson,
\newblock ``A {B}ayesian {A}pproach to {S}tochastic {R}oot {F}inding,''
\newblock in {\em Proceedings of the Winter Simulation Conference}, Phoenix,
  AZ, December 2011.

\bibitem{Han:2016}
W.~Han, P.~Rajan, P.~I. Frazier, and B.~M. Jedynak,
\newblock ``Probabilistic {G}roup {T}esting {U}nder {S}um {O}bservations: {A}
  {P}arallelizable 2-{A}pproximation for {E}ntropy {L}oss,''
\newblock {\em IEEE Transactions on Information Theory, to appear}, 2016.

\bibitem{Sznitman:2010}
R.~Sznitman and B.~Jedynak,
\newblock ``Active {T}esting for {F}ace {D}etection and {L}ocalization,''
\newblock {\em IEEE Transactions on Pattern Analysis and Machine Intelligence},
  vol. 32, no. 10, pp. 1914--1920, 2010.

\bibitem{Tsiligkaridis:2014}
T.~Tsiligkaridis, B.~M. Sadler, and A.~O.~Hero III,
\newblock ``Collaborative 20 {Q}uestions for {T}arget {L}ocalization,''
\newblock {\em IEEE Transactions on Information Theory}, vol. 60, no. 4, pp.
  2233--2352, April 2014.

\bibitem{Boyd:2006}
S.~Boyd, A.~Ghosh, B.~Prabhakar, and D.~Shah,
\newblock ``Randomized {G}ossip {A}lgorithms,''
\newblock {\em IEEE Transactions on Information Theory}, vol. 52, no. 6, pp.
  2508--2530, June 2006.

\bibitem{Dimakis:2010}
A.~Dimakis, S.~Kar, J.~M.~F. Moura, M.~G. Rabbat, and A.~Scaglione,
\newblock ``Gossip {A}lgorithms for {D}istributed {S}ignal {P}rocessing,''
\newblock {\em Proceedings of the IEEE}, vol. 98, no. 11, pp. 1847--1864,
  November 2010.

\bibitem{Kar:2011}
S.~Kar and J.~M.~F. Moura,
\newblock ``Convergence {R}ate {A}nalysis of {D}istributed {G}ossip ({L}inear
  {P}arameter) {E}stimation: {F}undamental {L}imits and {T}radeoffs,''
\newblock {\em IEEE Journal of Selected Topics in Signal Processing}, vol. 5,
  no. 4, pp. 674--690, August 2011.

\bibitem{Kar:2012}
S.~Kar, J.~M.~F. Moura, and K.~Ramanan,
\newblock ``{D}istributed {P}arameter {E}stimation in {S}ensor {N}etworks:
  {N}onlinear {O}bservation {M}odels and {I}mperfect {C}ommunication,''
\newblock {\em IEEE Transactions on Information Theory}, vol. 58, no. 6, pp.
  3575--3605, June 2012.

\bibitem{Lalitha:2014}
A.~Lalitha, T.~Javidi, and A.~Sarwate,
\newblock ``{Social Learning and Distributed Hypothesis Testing},''
\newblock {\em ArXiv e-prints}, Oct. 2014.

\bibitem{Nedic:2015}
A.~Nedic, A.~Olshevsky, and C.~A. Uribe,
\newblock ``Nonasymptotic {C}onvergence {R}ates for {C}ooperative {L}earning
  over {T}ime-{V}arying {D}irected {G}raphs,''
\newblock in {\em Proceedings of the American Control Conference (ACC)},
  Chicago, IL, July 2015.

\bibitem{Horstein:1963}
M.~Horstein,
\newblock ``Sequential {T}ransmission {U}sing {N}oiseless {F}eedback,''
\newblock {\em IEEE Transactions on Information Theory}, vol. 9, no. 3, pp.
  136--143, July 1963.

\bibitem{Jedynak:2012}
B.~Jedynak, P.~I. Frazier, and R.~Sznitman,
\newblock ``Twenty {Q}uestions with noise: {B}ayes optimal policies for entropy
  loss,''
\newblock {\em Journal of Applied Probability}, vol. 49, pp. 114--136, 2012.

\bibitem{Waeber:2013}
R.~Waeber, P.~I. Frazier, and S.~G. Henderson,
\newblock ``Bisection {S}earch with {N}oisy {R}esponses,''
\newblock {\em Journal of Control Optimization}, vol. 53, no. 3, pp.
  2261--2279, 2013.

\bibitem{Tsiligkaridis:2015}
T.~Tsiligkaridis, B.~M. Sadler, and A.~O.~Hero III,
\newblock ``On {D}ecentralized {E}stimation with {A}ctive {Q}ueries,''
\newblock {\em IEEE Transactions on Signal Processing}, vol. 63, no. 10, pp.
  2610--2622, May 2015.

\bibitem{Tsiligkaridis:2016}
T.~Tsiligkaridis,
\newblock ``Asynchronous {D}ecentralized {A}lgorithms for the {N}oisy 20
  {Q}uestions {P}roblem,''
\newblock in {\em Proceedings of the IEEE International Symposium on
  Information Theory (ISIT)}, Barcelona, Spain, July 2016.

\bibitem{Cover:2006}
T.~D. Cover and J.~A. Thomas,
\newblock {\em Elements of Information Theory},
\newblock Wiley, New York, NY, USA, 2006.

\bibitem{Hoel:1972}
P.~G. Hoel, S.~C. Port, and C.~J. Stone,
\newblock {\em Introduction to {S}tochastic {P}rocesses},
\newblock Houghton Mifflin Company, Boston, MA, USA, 1972.

\bibitem{Gupta:2000}
P.~Gupta and P.~R. Kumar,
\newblock ``The {C}apacity of {W}ireless {N}etworks,''
\newblock {\em IEEE Transactions on Information Theory}, vol. 46, no. 2, pp.
  388--404, March 2000.

\end{thebibliography}

\end{document}